\documentclass[authorcolumns,numberwithinsect,final]{no-lipics}

\usepackage{bm}
\usepackage{booktabs}
\usepackage[T1]{fontenc}

\SetKwComment{Comment}{$\triangleright$\ \rm}{}

\usetikzlibrary{arrows,arrows.meta,shapes.misc, decorations.pathreplacing,decorations.pathmorphing,calligraphy,patterns.meta}
\tikzset{dotmark/.style={circle,fill,inner sep=1.5pt}}
\tikzset{emptymark/.style={circle,draw,fill=white,inner sep=1.5pt}}
\tikzset{crossmark/.style={thick,inner sep=1.5pt}}

\usepackage{csquotes}
\usepackage[noend,ruled,linesnumbered]{algorithm2e}
\usepackage{amsmath}
\usepackage{amssymb}
\newcommand{\Z}{\mathbb{Z}}
\newcommand{\Zz}{\Z_{\ge 0}}

\newcommand{\Zp}{\Z_{>0}}
\newcommand{\Oh}{\mathcal{O}}
\newcommand{\bigO}{\Oh}
\newcommand{\dd}{.\,.}
\newcommand{\sub}{\subseteq}
\newcommand{\sm}{\setminus}
\newcommand{\T}{T}
\newcommand{\G}{\mathcal{G}} 
\newcommand{\Symb}{\mathcal{A}}
\newcommand{\Act}{\mathcal{B}}
\newcommand{\Left}{\mathcal{L}}
\newcommand{\Right}{\mathcal{R}}
\newcommand{\rle}{\mathsf{rle}}
\newcommand{\shrink}{\mathsf{shrink}}
\newcommand{\pc}{\mathsf{pc}}
\newcommand{\Pres}{\mathcal{S}}
\newcommand{\PSeq}{\mathsf{PSeq}}
\newcommand{\N}{\mathcal{N}}
\renewcommand{\S}{\mathcal{S}}
\newcommand{\Tr}{\mathcal{T}}
\newcommand{\uTr}{\overline{\mathcal{T}}}
\newcommand{\rhs}{\mathsf{rhs}}

\newcommand{\hX}{\bar{X}}
\newcommand{\emptystring}{\varepsilon}
\newcommand{\lvl}{\mathsf{lvl}}
\newcommand{\Occ}{\mathsf{Occ}}
\newcommand{\LCE}{\mathsf{LCE}}

\newcommand{\len}[1]{\| #1 \|}
\newcommand{\nil}{\mathbf{nil}}

\newcommand{\symb}[1]{#1.\mathsf{symb}}
\newcommand{\pos}[1]{#1.\mathsf{pos}}
\newcommand{\nexp}[1]{\exp(#1)}

\newcommand{\parent}[1]{#1.\mathtt{parent}}
\newcommand{\leaf}[1]{\mathtt{leaf}(#1)}
\newcommand{\descend}[2]{#1.\mathtt{leaf}(#2)}
\newcommand{\idx}[2]{#1.\mathtt{index}(#2)}
\newcommand{\child}[2]{#1.\mathtt{child}(#2)}
\newcommand{\sibling}[2]{#1.\mathtt{sibling}(#2)}
\newcommand{\level}[1]{#1.\mathsf{level}}

\newcommand{\nxt}[1]{#1.\mathtt{next}()}
\newcommand{\prv}[1]{#1.\mathtt{prev}()}

\def\fragmentco#1#2{\bm{[}\,#1\,\bm{.\,.}\,#2\,\bm{)}}
\def\fragmentoc#1#2{\bm{(}\,#1\,\bm{.\,.}\,#2\,\bm{]}}

\def\fragmentcc#1#2{\bm{[}\,#1\,\bm{.\,.}\,#2\,\bm{]}}
\def\position#1{\bm{[}\,#1\,\bm{]}}

\newtheorem{observation}[theorem]{Observation}

\usetikzlibrary{arrows.meta}

\title{Logarithmic-Time Internal Pattern Matching Queries in~Compressed and Dynamic Texts}
\author{Anouk Duyster}{Max Planck Institute for Informatics, SIC, Saarbrücken, Germany}{aduyster@mpi-sws.org}{https://orcid.org/0009-0009-6027-9377}{}
\author{Tomasz Kociumaka}{Max Planck Institute for Informatics, SIC, Saarbrücken, Germany}{tomasz.kociumaka@mpi-inf.mpg.de}{https://orcid.org/0000-0002-2477-1702}{Part of this work was done while at INSAIT, Sofia University "St. Kliment Ohridski", Bulgaria, funded from the Ministry of Education and Science of Bulgaria (support for INSAIT, part of the Bulgarian National Roadmap for Research Infrastructure).}

\relatedversion{A preliminary version of this work was presented at SPIRE 2024~\cite{DK24}.
This version contains several minor simplifications and describes our algorithm in more detail, consistently with our Rust implementation.}
\supplement{A Rust implementation of the proposed algorithm is provided at \url{https://github.com/aduyster/ipm_in_log_time}.}

\authorrunning{A. Duyster and T. Kociumaka}

\begin{document}
\maketitle
\setcounter{page}{1}
\begin{abstract}
    Internal Pattern Matching (IPM) queries on a text $T$, given two fragments $X$ and $Y$ of $T$ such that $|Y|<2|X|$, ask to compute all exact occurrences of $X$ within $Y$.
    IPM queries have been introduced by Kociumaka, Radoszewski, Rytter, and Waleń [SODA'15 \& SICOMP'24], who showed that they can be answered in $\Oh(1)$ time using a data structure of size $\Oh(n)$ and used this result to answer various queries about fragments of~$T$.

    In this work, we study IPM queries on compressed and dynamic strings.
    Our result is an $\Oh(\log n)$-time query algorithm applicable to any balanced recompression-based run-length straight-line program (RLSLP).
    In particular, one can use it on top of the RLSLP of Kociumaka, Navarro, and Prezza [IEEE TIT'23], whose size $\Oh\big(\delta \log \frac{n\log \sigma}{\delta \log n}\big)$ is optimal (among all text representations) as a function of the text length $n$, the alphabet size $\sigma$, and the substring complexity $\delta$.
    Our procedure does not rely on any preprocessing of the underlying RLSLP, which makes it readily applicable on top of the dynamic strings data structure of Gawrychowski, Karczmarz, Kociumaka, Łącki and Sankowski [SODA'18], which supports fully persistent updates in logarithmic time with high probability.
\end{abstract}

\section{Introduction}
Given two fragments $X=T\fragmentco{x}{x'}$ and $Y=T\fragmentco{y}{y'}$ of a text $T$ such that $|Y|<2|X|$,\footnote{The restriction $|Y|<2|X|$ ensures that the starting positions of the reported fragments form an arithmetic progression and thus can be represented in constant space.} an Internal Pattern Matching (IPM) query reports all exact occurrences of $X$ within $Y$, i.e., all fragments matching $X$ and contained within~$Y$.
Kociumaka, Radoszewski, Rytter, and Waleń~\cite{KRRW15} introduced IPM queries as a central building block for answering various further queries about fragments of~$T$.
They showed that, for every text of length $n$, IPM queries can be answered in $\Oh(1)$ time using a data structure of size $\Oh(n)$ that can be constructed in $\Oh(n)$ expected time.
The journal version of their work~\cite{KRRW23} provides an improved data structure that takes $\Oh(n  / \log_\sigma n)$ space and can be deterministically constructed in $\Oh(n/ \log_\sigma n)$ time if the characters of $T$ are integers in $\fragmentco{0}{\sigma}$ for some $\sigma = n^{\Oh(1)}$.
Hence, in the standard setting, IPM queries admit a compact data structure with optimal query and construction time.

Due to their multiple applications (see~\cite[Sections 1.2--1.3]{KRRW23}), IPM queries have also been considered in further settings, including the compressed setting, where the goal is to exploit compressibility of the text, and the dynamic setting, where the text may change over time.
Along with the Longest Common Extension (LCE) queries~\cite{LV88}, IPM queries constitute an elementary operation of the PILLAR model, introduced by Charalampopoulos, Kociumaka, and Wellnitz~\cite{CKW20} with the aim of unifying approximate pattern-matching
algorithms across different settings.
Whereas LCE queries can be answered $\Oh(\log n)$ time in the compressed~\cite{I17,KK23} and dynamic settings~\cite{ABR00,GKKLS18}, the algorithms for IPM queries have been slower so far.
In this work, we propose a novel procedure that answers IPM queries in logarithmic time both in the compressed and dynamic settings,
and thus we eliminate the bottleneck in the state-of-the-art PILLAR model implementations and multiple approximate pattern matching algorithms~\cite{CKPRRWZ22,CKW20,CKW22,CPRRWZ24,CGKMU22}.

\paragraph*{Our result}
State-of-the-art implementations of LCE and IPM queries represent compressed and dynamic texts using a run-length straight-line program (RLSLP) constructed using a locally consistent parsing scheme.
Such a scheme repeatedly partitions the text into blocks and replaces blocks with individual symbols (so that identical blocks are replaced by the same symbol) until the text consists of a single symbol.
The schemes alternate between run-length encoding (with blocks consisting of multiple copies of the same symbol) in odd rounds and another partitioning method (that reduces the text length by a constant factor while ensuring that matching fragments are partitioned consistently, except for blocks at the endpoints) in even rounds.
At each even round, the \emph{recompression} technique of Jeż~\cite{J15,J16} classifies the symbols into \emph{left} and \emph{right} symbols and creates a length-two block out of every left symbol followed by a right symbol; the remaining blocks are of length one.
A more recent \emph{restricted} variant~\cite{KNP23,KRRW23} further classifies some symbols as \emph{inactive} so that they always form length-one blocks.

In this work, we show how to efficiently answer IPM queries using any RLSLP obtained via a (possibly restricted) recompression scheme.
Our algorithm does not need any processing of the recompression RLSLP; it only assumes that every non-terminal stores its production, the length of its expansion, and the index of the round when it has been created.

\begin{theorem}\label{thm:main}
IPM queries on a text $T$  represented using a (restricted) $r$-round recompression run-length straight-line program can be answered in $\Oh(r)$ time.\lipicsEnd
\end{theorem}

\paragraph*{IPM queries in compressed texts}
Kociumaka, Navarro, and Prezza~\cite{KNP23} proved that every non-empty text $T$ admits an $\Oh(\log n)$-round restricted recompression RLSLP of size $\Oh\big(\delta \log\frac{n\log \sigma}{\delta \log n}\big)$, where $n$ is the length of the text, $\sigma$ is the alphabet size, and $\delta$ is the substring complexity of $T$.
More recently, Kempa and Kociumaka~\cite{KK23} showed how to construct such an RLSLP in $\Oh(\delta\log^7 n)$ time from the Lempel--Ziv (LZ77) parsing of $T$.
Combining this with Theorem~\ref{thm:main}, we obtain the following:

\begin{corollary}\label{cor:compressed}
For every non-empty text $T$, there is a data structure of size $\Oh\big(\delta \log\frac{n\log \sigma}{\delta \log n}\big)$ that answers IPM queries in $\Oh(\log n)$ time.
Such a data structure can be constructed in $\Oh(\delta \log^7 n)$ time given the LZ77 parsing of~$T$.\lipicsEnd
\end{corollary}

Analogous results are already known for LCE queries \cite[Theorem 5.25]{KK23}, which means that the PILLAR model can be implemented in $\Oh(\log n)$ time per operation using $\Oh(\delta \log\frac{n\log \sigma}{\delta \log n})$ space and $\Oh(\delta \log^7 n)$ preprocessing time.\footnote{Using the results of \cite{I17} instead of \cite{KK23}, one can achieve a faster $\Oh(g \log \frac{n}{g})$-time construction from a size-$g$ straight-line grammar generating $T$ at the expense of a larger data structure of size $\Oh(z \log \frac{n}{z})$, where $z$ is the size of the LZ77 parsing of $T$.}
Prior to this work, state-of-the-art implementations of IPM queries in the compressed setting required $\Oh(\log^3 n)$ time~\cite{KK20} or $\Oh(\log^2 n \log \log n)$ time~\cite{CKW20} per query.

\paragraph*{IPM queries in dynamic texts}
The dynamic strings data structure of Gawry\-chowski, Karczmarz, Kociumaka, Łącki, and Sankowski \cite{GKKLS18} maintains a collection of non-empty strings, allowing one to add new strings to the collection explicitly ($\texttt{make\_string}$), by non-destructively concatenating two existing strings ($\texttt{concat}$), or by non-destructively splitting an existing string into two pieces ($\texttt{split}$).
These updates take $\Oh(\log N)$ time with high probability, where $N$ is the total length of the strings of the collection, except for $\texttt{make\_string}$, which takes $\Oh(n + \log N)$ time for creating a string of length $n$.
Internally, the strings are represented using an $r$-round (non-restricted) recompression RLSLP, where $r=\Oh(\log N)$ with high probability.
Hence, Theorem~\ref{thm:main} yields the following
\begin{corollary}\label{cor:dynamic}
The dynamic strings data structure of~\cite{GKKLS18} supports IPM queries in $\Oh(\log N)$ time with high probability, where $N$ is the total length of the stored strings.\lipicsEnd
\end{corollary}
The dynamic strings data structure supports LCE queries in $\Oh(\log N)$ time w.h.p., so this yields a PILLAR model implementation in $\Oh(\log N)$ time per operation.
Previous implementations supported IPM queries in $\Oh(\log^2 N)$ time~\cite{CKW20}.
Faster IPM queries in the dynamic setting were only known for an alternative deterministic dynamic strings implementation~\cite{KK22} that takes $\Oh(\log N \log^{2-o(1)} \log N)$ time per update as well as for both LCE and IPM queries; see~\cite{CKW22} for a discussion.

\section{Preliminaries}
In this section, we formally define restricted recompression and the underlying concepts and notations.
Compared to previous works~\cite{GKKLS18,KK23,KRRW23} and the preliminary version of this paper~\cite{DK24}, we provide much more details on how to efficiently traverse (uncompressed) parse trees.

\paragraph*{Basic definitions}
A string is a finite sequence of characters from a given alphabet $\Sigma$.
The length of a string $S$ is denoted $|S|$, and the empty string is denoted $\emptystring$.
For a \emph{position} $i\in \fragmentco{0}{|S|}$,\footnote{For $a,b\in \mathbb{R}$, denote $\fragmentcc{a}{b}=\{k\in \Z : a \le k \le b\}$, $\fragmentco{a}{b} = \{k\in \Z : a \le k < b\}$, and $\fragmentoc{a}{b} = \{k \in \Z : a < k \le b\}$.} the $i$th character of $S$ is denoted $S\position{i}$.
A string $U$ is a \emph{substring} of a string $S$ if $U=S\position{i}S\position{i+1}\cdots S\position{j-1}$ holds for some integers $0\le i \le j \le |S|$.
In this case, we say that $U$ \emph{occurs} in $S$ at position $i$, and we denote the \emph{occurrence} of $U$ at position $i$ of $S$ by $S\fragmentco{i}{j}$.
We call $S\fragmentco{i}{j}$ a \emph{fragment} of $S$ and, in some contexts, also denote it by $S\fragmentcc{i}{j-1}$ or $S\fragmentoc{i-1}{j-1}$. Fragments  the form $S\fragmentco{0}{j}$ and $S\fragmentco{i}{|S|}$ are \emph{prefixes} and \emph{suffixes} of $S$, respectively.

Technically, a fragment $S\fragmentco{i}{j}$ is interpreted as a tuple $(S, i, j)$ consisting of the positions $i$ and $j$ as well as (a constant reference to) the string $S$.
We say that fragments $S\fragmentco{i}{j}$ and $T\fragmentco{i'}{j'}$ of strings $S$ and $T$, respectively, \emph{match} if they are occurrences of the same substring.
Consistently with the literature, we then write $S\fragmentco{i}{j}=T\fragmentco{i'}{j'}$ even though the triples $(S,i,j)$ and $(T,i',j')$ are not necessarily equal.

We denote by $UV$ or $U\cdot V$ the concatenation of two strings $U$ and $V$, that is, $UV = U\position{0} \cdots U\position{|U|-1}\allowbreak V\position{0} \cdots V\position{|V|-1}$.
Furthermore, both $S_0\cdots S_{k-1}$ and $\bigodot_{i=0}^{k-1} S_i$ denote the concatenation of $k\in \Zz$ strings $S_0,\ldots,S_{k-1}$.
Moreover, $S^k=\bigodot_{i=0}^{k-1} S$ denotes the concatenation of $k$ copies of the same string $S$.

We use $\overline{S}$ to denote the reverse of $S$, that is, $\overline{S}=S\position{|S|-1} \cdots S\position{0}$.
An integer $p \in \fragmentoc{0}{|S|}$ is a period of a string $S$ if $S\position{i} = S\position{i+p}$ holds for every $i \in \fragmentco{0}{|S|-p}$; equivalently, $p$ is a period of $S$ if and only if the length-$(|S|-p)$ prefix and suffix of $S$ match, that is, $S\fragmentco{0}{|S|-p}=S\fragmentco{p}{|S|}$.

\paragraph*{Straight-line grammars}
For a fixed context-free grammar $\G$, we denote by $\Sigma$ and $\N$ the sets of terminals and non-terminals, respectively.
The set of \emph{symbols} is $\S:=\Sigma\cup \N$.
We say that $\G$ is a~\emph{straight-line grammar} (SLG) if:
\begin{itemize}
\item each non-terminal $A\in \N$ has a unique production $A\to \rhs(A)$, whose right-hand side is a non-empty sequence of symbols, i.e., $\rhs(A)\in \S^+$, and
\item the set of symbols $\S$ admits a partial order $\prec$ such that $B \prec A$ if $B$ occurs in $\rhs(A)$.
\end{itemize}
A straight-line grammar $\G$ is a \emph{run-length straight-line program} (RLSLP) if each production \(A \to \rhs(A)\) is of one of the following two types:
\begin{description}
\item[Pair:\hspace{.35cm}] $\rhs(A)= BC$ for some symbols $B,C\in \S$ such that $B \ne C$;
\item[Power:] $\rhs(A) = B^m$ for a symbol $B\in \S$ and an integer $m \ge 2$.
\end{description}

Every straight-line grammar $\G$ yields an \emph{expansion} function $\exp: \S^* \to \Sigma^*$ assigning to every string $A\in \S^*$ the unique terminal-symbol string $\exp(A)\in \Sigma^*$ derivable from $A$.
The function $\exp$ also admits a concise recursive definition:
\[\exp(A) = \begin{cases}
  A & \text{if $A\in \Sigma$},\\
  \exp(\rhs(A)) & \text{if $A\in \N$},\\
  \bigodot_{i=0}^{a-1}\exp(A\position{i}) & \text{if $A\in \S^{a}$ for $a\ne 1$.}
\end{cases}\]
The expansion of the starting symbol of $\G$ is the string \emph{represented} by~$\G$.

\paragraph*{Restricted recompression}
Both recompression and restricted recompression, given a string $\T\in \Sigma^*$, construct a sequence of strings $(\T_k)_{k=0}^\infty$ over an infinite alphabet $\Symb$ defined as the least fixed point of the following equation:
 \[\Symb = \Sigma \cup (\Symb \times \Symb)\cup (\Symb \times \mathbb{Z}_{\ge 2} ).\]
In other words, $\Symb = \bigcup_{k=0}^\infty \Symb_k$, where $\Symb_0=\Sigma$ and $\Symb_{k}=\Symb_{k-1}\cup (\Symb_{k-1}\times \Symb_{k-1})\cup (\Symb_{k-1} \times \mathbb{Z}_{\ge 2})$ for $k\in \Zp$.
Symbols in $\Symb \sm \Sigma$ are non-terminals
with $\rhs((B,C))=BC$ for $(B,C)\in \Symb \times \Symb$ and $\rhs((B,m))= B^m$
for $(B,m)\in \Symb\times \mathbb{Z}_{\ge 2}$.
Intuitively, $\Symb$ can be interpreted as a universal RLSLP: for every RLSLP with symbols $\S$ and terminals $\Sigma \subseteq \S$, there is a unique homomorphism $f : \S^* \to \Symb^*$ such that $f(A)$ = $A$ if $A \in \Sigma$,
$\rhs(f(A)) = f(\rhs(A))$ if $A \in \S \setminus \Sigma$, and $f(A)=\bigodot_{i=0}^{|A|-1} f(A\position{i})$ for every $A \in \S^*$.
As a result, $\Symb$ provides a convenient formalism for reasoning about procedures generating RLSLPs.

Restricted recompression uses the following transformations of $\Symb^*$ to $\Symb^*$.

\begin{definition}[Restricted run-length encoding~\cite{KRRW23,KNP23}]\label{def:rle}
Given $\T\in \Symb^*$ and $\Act \sub \Symb$, we define $\rle_{\Act}(\T)\in \Symb^*$
to be the string obtained as follows by decomposing $\T$ into blocks and collapsing these blocks:
\begin{enumerate}
  \item For $i\in \fragmentco{0}{|T|-1}$, place a \emph{block boundary} between $\T\position{i}$ and $\T\position{i+1}$
  unless $\T\position{i}=\T\position{i+1}\in \Act$.
  \item Replace each block $\T\fragmentco{i}{i+m}= A^m$ of length $m\ge 2$ with a symbol $(A,m)$.\qedhere
\end{enumerate}
\end{definition}
\begin{remark}
By construction, $\exp(\rle_\Act(\T))=\exp(T)$ holds for every $\T\in \Symb^*$ and $\Act \sub \Symb$.
\end{remark}

\begin{definition}[Restricted pair compression~\cite{KRRW23,KNP23}]\label{def:pc}
  Given $\T\in \Symb^*$ and disjoint sets $\Left,\Right \sub \Symb$, we define $\pc_{\Left,\Right}(\T)\in \Symb^*$
  to be the string obtained as follows by decomposing $\T$ into blocks and collapsing these blocks:
  \begin{enumerate}
    \item For $i\in \fragmentco{0}{|T|-1}$, place a \emph{block boundary} between $\T\position{i}$ and $\T\position{i+1}$
    unless $\T\position{i}\in \Left$ and $\T\position{i+1}\in \Right$.
    \item Replace each block $\T\fragmentcc{i}{i+1}$ with a symbol $(\T\position{i},\T\position{i+1})$.\qedhere
  \end{enumerate}
\end{definition}
\begin{remark}
By construction, $\exp(\pc_{\Left,\Right}(\T))=\exp(\T)$ holds for every $\T\in \Symb^*$ disjoint $\Left,\Right \sub \Symb$.
\end{remark}

Restricted recompression repeatedly transforms the input text alternating between restricted run-length encoding and restricted pair compression.
Unlike in the classic (non-restricted) recompression~\cite{J15,J16}, it is allowed to set $\Act\subsetneq \Symb$ and $\Left\cup\Right\subsetneq \Symb$.

\begin{definition}[Restricted recompression~\cite{KRRW23,KNP23}]\label{def:recompression}
For a string $\T\in \Sigma^*$, we define $T_0=T$ and $T_k=\shrink_{k}(T_{k-1})$ for every $k\in \Zp$,
where
\[\shrink_k = \begin{cases} \rle_{\Act_k} & \text{for some fixed set $\Act_k\sub \Symb$ if $k$ is odd,}\\
  \pc_{\Left_k,\Right_k} & \text{for some fixed disjoint sets $\Left_k,\Right_k\sub\Symb$ if $k$ is even.}\end{cases}\]
If $r=\min\{k\in \Zz : |\T_k|=1\}$ exists,\footnote{Constructions in~\cite{J15,KRRW23,KNP23} provide several strategies of picking $\Act_k$, $\Left_k$, and $\Right_k$ based on a given text $T\in \Sigma^n$ so that $|T_r|=1$ holds for some $r=\Oh(\log n)$.
There also randomized constructions~\cite{GKKLS18,KRRW23}, where $\Act_k$, $\Left_k$, and $\Right_k$ are random variables, and, for every $\T \in \Sigma^n$, with high probability, $|T_r|=1$ holds for some $r=\Oh(\log n)$.} we define the underlying \emph{$r$-round restricted recompression RLSLP} $\G$ to be an RLSLP generating $T$ with symbols $\Pres = \bigcup_{k=0}^r \Pres_k$, where $\Pres_k = \{\T_k\position{j} : j\in \fragmentco{0}{|T_k|}\}$, and starting symbol~$T_r\position{0}$.
\end{definition}

Our query algorithm assumes that the text $T$ is represented as an $r$-round restricted recompression RLSLP with every $A\in \N$ storing its right-hand side $\rhs(A)$, expansion length $\len{A}\coloneq |\exp(A)|$, and \emph{level} \[\lvl(A)\coloneq\min\{k\in \Zz : A\in \S_k\}.\]

\paragraph*{Parse tree}
The \emph{parse tree} $\Tr(A)$ of a symbol $A\in \S$ in a straight-line grammar is a rooted ordered tree with each node $\nu$ associated to a symbol $\symb{\nu}\in \S$.
The root of $\Tr(A)$ is a node $\rho$ with $\symb{\rho}=A$.
If $A \in \Sigma$, then $\rho$ has no children.
If $A\in \N$ and $\rhs(A)=A_0\cdots A_{a-1}$, then $\rho$ has $a$ children,
and the subtree rooted at the $i$th child is (a copy of) $\Tr(A_i)$.
The parse tree $\Tr$ of a straight-line grammar $\G$ is the parse tree of the starting symbol of $\G$ (whose expansion is the text $T$ represented by $\G$).

Each node $\nu$ of $\Tr$ is also associated with a position $\pos{\nu}\in \fragmentco{0}{|T|}$, which we interpret as the index of the leftmost leaf in the subtree of $\Tr$ rooted at $\nu$.
Moreover, we denote by $\nexp{\nu}\coloneq T\fragmentco{\pos{\nu}}{\pos{\nu}+\len{\symb{\nu}}}$ the fragment of $\T$ corresponding to $\nu$.
For the root $\rho$, we define  $\pos{\rho}=0$ so that $\nexp{\rho}=T\fragmentco{0}{|T|}$ is the whole $T$.
Moreover, if $\rhs(\symb{\nu})=A_0\cdots A_{a-1}$ and $\nu_{0},\ldots,\nu_{a-1}$ are the children of $\nu$, then $\pos{\nu_i} = \pos{\nu} + \sum_{j=0}^{i-1} \len{A_j}$ for $i\in \fragmentco{0}{a}$.
This way, $\nexp{\nu}$ is the concatenation of fragments $\nexp{\nu_0},\ldots,\nexp{\nu_{a-1}}$.
A straightforward inductive argument shows that $\nexp{\nu}$ is an occurrence of $\exp(\symb{\nu})$ in $\Tr$.

Although we typically cannot afford to store the parse tree $\Tr$ explicitly, our representation of the restricted recompression RLSLP allows for efficiently traversing the parse tree~$\Tr$ using an abstraction of \emph{pointers}.
We implement a pointer to a node $\nu\in \Tr$ as a tuple $(\pos{\nu},\symb{\nu},\parent{\nu})$,
where $\parent{\nu}$ is (a pointer to) the parent $\nu$ (or $\nil$ if $\nu$ is the root of~$\Tr$).
The symbol $\symb{\nu}$ and the parent $\parent{\nu}$ are stored as constant references (allowing many references to the same immutable object in the computer memory).
This way, one can think of a pointer as a stack containing $\pos{\mu}$ and $\symb{\mu}$ for every ancestor $\mu$ of $\nu$ starting from the root $\rho$ at the bottom of the stack to the node $\nu$ itself at the top of the stack.
The stack implementation is like in functional programming languages so that multiple pointers can share a common prefix of the stack.

Given (a pointer to) a node $\nu$, we can in constant time retrieve (a pointer to) the parent $\parent{\nu}$.
Given additionally an index $i$, we can also construct (a pointer to) the $i$th child of $\nu$, denoted $\child{\nu}{i}$; we assume that $\child{\nu}{i}=\nil$ for out-of-bounds indices $i$.
Additionally, given any position $j\in \fragmentco{\pos{\nu}}{\pos{\nu}+\len{\symb{\nu}}}$, we can compute an index $i\coloneq \idx{\nu}{j}$ of the unique child $\nu_i$ of $\nu$ such that $j\in \fragmentco{\pos{\nu_i}}{\pos{\nu_i}+\len{\symb{\nu_i}}}$.
Using this functionality, we can descend to the desired leaf of $\Tr$:

\begin{observation}\label{obs:leaf}
  Given any position $j\in \fragmentco{0}{|T|}$, a pointer to the $j$-th leaf of $\Tr$, that is, a node $\leaf{j}\in \Tr$ such that $\exp(\leaf{j})=T\position{j}$, can be retrieved in $\Oh(r)$ time.
\end{observation}
\begin{proof}
  We define a more general recursive $\descend{\nu}{j}$ function that computes $\leaf{j}$ given (a pointer to) a node $\nu$ such that $j\in \fragmentco{\pos{\nu}}{\pos{\nu}+\len{\symb{\nu}}}$.
  If $\nu$ is already a leaf (i.e., $\len{\symb{\nu}}=1$), then $\descend{\nu}{j}=\nu$.
  Otherwise, $\descend{\nu}{j} = \descend{\child{\nu}{\idx{\nu}{j}}}{j}$.
  Finally, we observe that $\leaf{j}$ can be computed as $\descend{\rho}{j}$, where $\rho$ is the root of $\Tr$.

  The time complexity is proportional to the length of the path from $\rho$ to $\leaf{j}$, which is $\Oh(r)$.
\end{proof}

If $\mu = \parent{\nu}$ and $\mu_0,\ldots,\mu_{a-1}$ are the children of $\mu$, then we can retrieve the index $i$ such that $\nu = \mu_i$ as $i\coloneqq \idx{\parent{\nu}}{\pos{\nu}}$.
For $d\in \fragmentco{-i}{a-i}$, we refer to $\mu_{i+d}$ as the $d$th sibling of $\nu$ and, for convenience, write $\sibling{\nu}{d}\coloneq \child{\parent{\nu}}{d+\idx{\parent{\nu}}{\pos{\nu}}}$.
We assume that $\sibling{\nu}{d}=\nil$ if $d\notin \fragmentco{-i}{a-i}$ or $\mu= \parent{\nu}=\nil$.

\paragraph*{Uncompressed parse tree}
\newcommand{\unu}{\overline\nu}
\newcommand{\umu}{\overline\mu}

We also observe that each symbol in each of the strings $(T_k)_{k=0}^\infty$ can be associated with a node in the parse tree $\Tr$.
This mapping is not injective, though: if a symbol of $T_k$ forms a length-one block with respect to $\shrink_{k+1}$, it remains in $T_{k+1}$ represented by the same node of the parse tree $\Tr$.
As in \cite{GKKLS18}, this can be alleviated using an \emph{uncompressed parse tree} $\uTr$,
where the edges of $\Tr$ are subdivided so that a node $\nu\in \Tr$ gets a separate copy for each level $k$ such that $T_k$ has a symbol associated with $\nu$.
As a result, there is a bijection between the characters $T_k\position{i}$, where $k\in \Zz$ and $i\in \fragmentco{0}{|T_k|}$, and the nodes $\unu$ of $\uTr$.\footnote{To support arbitrarily large $k\in \Zz$, the root of $\Tr$ corresponds to an infinite path in $\uTr$, with one node per level $k\ge r$.}

For each node $\unu\in \uTr$, we denote $\level{\unu}$ and $\symb{\unu}$ to be the level of $\unu$ and the symbol represented by $\unu$, respectively; in other words, $\level{\unu}=k$ and $\symb{\unu}=A$ if and only if $\unu$ represents $T_k\position{i}=A$ for some $i\in \fragmentco{0}{|T_k|}$.
Note that $\level{\unu} \ge \lvl(\symb{\unu})$, but the equality does not necessarily hold.
\begin{fact}\label{fct:expand_level}
  A node $\unu\in \uTr$ has exactly one child if and only if $\level{\unu} > \lvl(\symb{\unu})$.
\end{fact}
\begin{proof}
  Suppose that $\unu$ represents $T_k\position{i}=A$ for some $i\in \fragmentco{0}{|T_k|}$, where $\level{\unu}=k$ and $\symb{\unu}=A$.
  Moreover, let $\ell = \lvl(A)$.
  If $k=\ell=0$, then, trivially, $\unu$ has no children.
  If $k=\ell>0$, then, by definition of~$\ell$, the symbol $A$ cannot occur in $T_{k-1}$, so $T_{k}\position{i}$ must have been obtained by collapsing a block of $T_{k-1}$ (with respect to $\shrink_k$) that consists of $|\rhs(A)|\ge 2$ symbols, and thus $\unu$ has at least two children.

  Finally, for a proof by contradiction, suppose that $k>\ell$ yet $\unu$ has at least two children (it cannot have a single child since $k>0$ and every block with respect to $\shrink_k$ is non-empty).
  This means that $T_k\position{i}$ must have been obtained by collapsing a block $T_{k-1}\fragmentco{i'}{j'}$ (with respect to $\shrink_k$) of length $j'-i'\ge 2$.
  Further expansions of this block yield a fragment $T_{\ell}\fragmentco{i''}{j''}$ of length $j''-i''\ge j'-i'\ge 2$ whose expansion matches $\exp(T_k\position{i})=\exp(A)$.
  At the same time, by definition of $\ell$, the symbol $A$ occurs as one of the characters (length-$1$ fragments) of $T_\ell$.
  This contradicts \cite[Lemma 4.4]{KRRW23}, which states that two fragments of $T_\ell$ match if and only if they have matching expansions.
\end{proof}

We use $\parent{\unu}$ and $\child{\unu}{i}$ to denote the parent and the $i$th child of $\unu$, respectively (with the latter defined as $\nil$ for out-of-bounds arguments $i$).
Observe that every node $\unu\in \uTr$ satisfies $\level{\parent{\unu}}=\level{\unu}+1$ and, if $\child{\unu}{i} \ne \nil$, then also $\level{\child{\unu}{i}}=\level{\unu}-1$.
\cref{fct:expand_level} lets us easily implement the parent and child operations based on the analogous operations on $\Tr$:
\begin{observation}
  Given a (pointer to) a node $\unu\in \uTr$, one can in constant time compute (pointers to) $\parent{\unu}$ as well as $\child{\unu}{i}$ for any index $i\in \Z$.
\end{observation}
\begin{proof}
  Suppose that $\unu$ is represented as $(\nu,k)$, where $\nu$ is the corresponding node of $\Tr$ and $k=\level{\unu}$.
  If $\parent{\nu}\ne \nil$ (i.e., $\nu$ is not the root of $\Tr$) and $\lvl(\symb{\parent{\nu}})= k+1$, then, by \cref{fct:expand_level}, we have $\parent{\unu}=(\parent{\nu},k+1)$; otherwise, $\parent{\unu}=(\nu,k+1)$.

  Furthermore, $\child{\unu}{0}=(\nu,k-1)$ if $k>\lvl(\symb{\nu})$, and $\child{\unu}{i}=(\child{\nu}{i},k-1)$ if $k=\lvl(\symb{\nu})$ and $\child{\nu}{i}\ne \nil$. In the remaining cases, $\child{\unu}{i}=\nil$.
\end{proof}

If $\unu$ represents $T_k\position{i}$, then we use $\nxt{\unu}$ and $\prv{\unu}$ to denote nodes representing $T_k\position{i+1}$ and $T_k\position{i-1}$, respectively, with $\nxt{\unu}=\nil$ if $i=|T_k|-1$ and $\prv{\unu}=\nil$ if $i=0$.
A naive implementation of the operation computing $\umu=\nxt{\unu}$ traverses the path from $\unu$ to $\umu$ in the uncompressed parse tree~$\uTr$, and symmetrically for $\unu=\prv{\umu}$.
As shown in \cite{GKKLS18}, a more sophisticated implementation of pointers (which requires additional auxiliary data structures and an assumption that the machine word consists of $\Omega(r)$ bits) allows for a constant-time implementation.
In this work, we use a convenient middle-ground implementation based on the parse tree $\Tr$:
\begin{lemma}\label{lem:traverse}
The $\nxt{\cdot}$ operation on uncompressed parse tree nodes can be implemented so that,
for every node $\unu_0\in \uTr$, every sequence of $s$ valid operations $(\unu_i \coloneq \unu_{i-1}.\mathsf{op}_i)_{i=1}^s$ takes $\Oh(r+s)$ time provided that, for each $i\in \fragmentcc{1}{s}$, the operation $\unu_i \coloneq\unu_{i-1}.\mathsf{op}_i$ is of the form $\unu_i \coloneq  \nxt{\unu_{i-1}}$, $\unu_i \coloneq  \parent{\unu_{i-1}}$, or $\unu_i \coloneq  \child{\unu_{i-1}}{j}$ for some $j\in \Z$.
The $\prv{\cdot}$ operation satisfies the same condition.\footnote{If $\prv{\cdot}$ and $\nxt{\cdot}$ are both used in the sequence of $s$ operations, we do not guarantee the $\Oh(r+s)$ running time.}
\end{lemma}
\begin{proof}
Let us first describe the implementation of $\nxt{\unu}$ provided that $\unu\in \uTr$ is represented by $(\nu,k)$, where $\nu$ is the corresponding node of $\Tr$ and $k=\level{\unu}$.
Suppose that $\nxt{\unu}=\umu$, where $\umu\in \uTr$ is analogously represented by $(\mu,k)$.
Our algorithm traverses the path from $\nu$ to $\mu$ in $\Tr$.
For this, we first traverse the path from $\nu$ towards the root of $\Tr$ (using the $\parent{\cdot}$ operation) until we encounter a node that has a right sibling.
If this does not happen before we reach the root of $\Tr$, then $\nxt{\unu}=\nil$.
Otherwise, we proceed to the right sibling (using the $\sibling{\cdot}{1}$ operation) and then descend $\Tr$ (using the $\child{\cdot}{0}$ operation) until we reach $\mu$, i.e., the first node that we encounter that satisfies $\lvl(\symb{\mu})\le k$.
Our implementation of pointers to $\Tr$ guarantees that the running time of this procedure is proportional to the length of the path from $\nu$ to $\mu$ in $\Tr$.

Consider a sequence $(\unu_i \coloneq \unu_{i-1}.\mathsf{op}_i)_{i=1}^s$ of $s$ calls, and, for $i\in\fragmentcc{0}{s}$, denote by $\nu_i$ the nodes of $\Tr$ corresponding to $\unu_i$.
Observe that the paths from $\nu_{i}$ to $\nu_{i+1}$ form subsequent fragments of the Euler tour of the subtree of $\Tr$ consisting of all the nodes $\nu_0,\ldots,\nu_s$ and their ancestors.
Since each internal node in $\Tr$ has at least two children, the subtree size is $\Oh(s+r)$: besides the $s+1$ nodes  $\nu_0,\ldots,\nu_s$ and their at most $s$ lowest common ancestors, it only contains $\Oh(r)$ ancestors of $\nu_0$ and~$\nu_s$.

The implementation and analysis of the $\prv{\cdot}$ operation is symmetric.
\end{proof}

\section{Popped sequences}
Efficient implementation of LCE queries on recompression RLSLPs relies on a notion of \emph{popped sequences} introduced by I~\cite{I17}.
We use a slightly simpler construction that is not tailored to any specific version of recompression.

\begin{definition}[Popped Sequence]\label{def:popped}
Given a string $X\in \Sigma^*$, we define strings $L_k,R_k,\hX_k\in \Symb^*$ for $k\in \Zz$ as follows.
We set $\hX_0 = X$.
For $k\ge 0$, we define $L_{k}$, $R_{k}$, and $\hX_{k+1}$ depending on the blocks of $\shrink_{k+1}(\hX_k)$:
  \begin{itemize}
    \item $L_k$ is of the leftmost block of $\hX_k$ (with respect to $\shrink_{k+1}$) unless $\hX_k$ is partitioned into zero blocks (i.e., $\hX_k=\emptystring$) or the leftmost block of $\hX_k$ consists of two distinct symbols. In these two special cases, $L_k=\emptystring$.
    \item $R_k$ is of the rightmost block of $\hX_k$ (with respect to $\shrink_{k+1}$)
    unless $\hX_k$ is partitioned into at most one block or the rightmost block of $\hX_k$ consist of two distinct symbols. In these two special cases, $R_k=\emptystring$.
    \item $\hX_{k+1}=\shrink_{k+1}(\hX'_k)$, where $\hX'_k\in \Symb^*$ is the substring of $\hX_k$ such $\hX_k = L_k \cdot \hX'_{k} \cdot R_k$.
  \end{itemize}
The \emph{popped sequence} of $X$ is $\PSeq(X)\coloneq L_0\cdots L_q\cdot R_q \cdots R_0$ for any $q\in \Zz$ such that $\hX_{q+1}=\emptystring$.\end{definition}

In other words, to construct the sequence $(\hX_k)_{k=0}^\infty$, we proceed as in Definition~\ref{def:recompression},
but we remove the leftmost and the rightmost block of $\hX_k$ before we collapse the remaining blocks to obtain $\hX_{k+1}$, unless here are no blocks left to be removed or we would remove a block of two distinct symbols (we keep such a block and the symbol it collapses into will be removed at the next level, either on its own or as a part of a longer run).
The removed blocks are concatenated (in the appropriate order) to form $\PSeq(X)$.
The following observation captures the immediate consequences of the \cref{def:popped}.

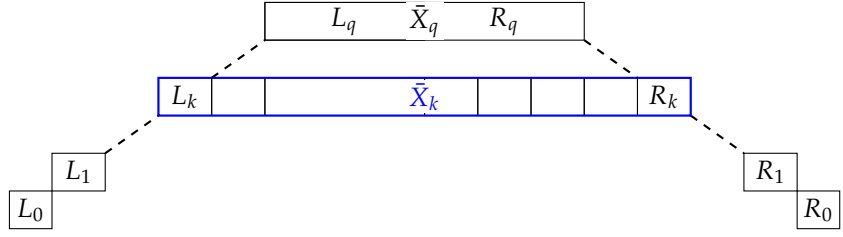
\begin{figure}
\center
\begin{tikzpicture}[yscale=0.5, xscale=0.7]
  \draw (0.2,0) rectangle (1,1);  \node at (0.6,0.5) {\(L_0\)};
  \draw (1,1) rectangle (2,2);  \node at (1.5,1.5) {\(L_1\)};

  \draw (14,1) rectangle (15,2); \node at (15.4,0.5) {\(R_0\)};
  \draw (15,1) rectangle (15.8,0); \node at (14.5,1.5) {\(R_1\)};

  \draw [thick,dashed] (2,2) -- (3,3);
  \draw [thick,dashed] (13,3) -- (14,2);

  \draw (3,3) rectangle (4,4);  \node at (3.5,3.5) {\(L_k\)};
  \draw (4,3) rectangle (5,4);
  \draw (5,3) rectangle (8,4);
  \draw (8,3) rectangle (9,4);
  \draw (9,3) rectangle (10,4);
  \draw (10,3) rectangle (11,4);
  \draw (11,3) rectangle (12,4);
  \draw (12,3) rectangle (13,4); \node at (12.5,3.5) {\(R_k\)};
  \draw[color=blue, thick] (3,3) rectangle (13,4);
  \node[color=blue,fill=white,inner sep=1pt] at (8,3.5) {\(\hX_k\)};

  \draw [thick,dashed] (4,4) -- (5,5);
  \draw [thick,dashed] (11,5) -- (12,4);

  \draw (5,5) rectangle (11,6);
  \draw (5,5) rectangle (8,6);
  \node at (6.5,5.5) {\(L_q\)};
  \node[fill=white,inner sep=1pt] at (8,5.5) {\(\hX_q\)};
  \node at (9.5,5.5) {\(R_q\)};
\end{tikzpicture}
\caption{The popped sequence is build from the blocks \(L_0\) through \(R_0\).
In every level \(k\), the string \(\hX_k\) spans from $L_k$ to $R_k$ (inclusive).}
\end{figure}

\begin{observation}\label{obs:one_run}
For every $k\in \Zz$, we have $|\shrink_{k+1}(L_k)|\le 1$, $|\shrink_{k+1}(R_k)|\le 1$,  $|\rle(L_k)|\le 1$ and $|\rle(R_k)|\le 1$, as well as $X=\exp(L_0\cdots L_{k-1}\cdot \hX_k \cdot R_{k-1} \cdots R_0)$.\lipicsEnd
\end{observation}

A crucial property of the popped sequence is that every occurrence of $X$ in a text $T$ induces the occurrences of each non-empty string $\hX_k$ in the string $T_k$, and thus an ``occurrence'' of the entire popped sequence in the parse tree of $T$.
\begin{lemma}\label{lem:occx}
If $X\in \Sigma^+$ occurs in a text $T$ at position $i$ and $\hX_k\ne \emptystring$, then $\hX_k$ occurs in $T_k$ at a position $i_k$ such that $|\exp(T_k\fragmentco{0}{i_k})|=i+|\exp(L_0\cdots L_{k-1})|$.
\end{lemma}
\begin{proof}
  We proceed by induction on $k$. The base case is trivial since $T_0=T$ and $\hX_0=X$.
  For $k\in \Zp$, the inductive hypothesis yields an occurrence of $\hX_{k-1}$ starting at position $i_{k-1}$.
  By Definitions~\ref{def:rle} and~\ref{def:pc}, $\shrink_{k}$ places block boundaries based on the identities of the two surrounding symbols.
  In particular, the block boundaries strictly inside $\hX_{k-1}$ are placed in the same way regardless of whether $\hX_{k-1}$ is processed as a standalone string or as a fragment of $T_{k-1}$.
  Moreover, if the leftmost block of $\hX_{k-1}$ consists of two distinct symbols, then there is a block boundary at the beginning of every occurrence of $\hX_{k-1}$ in $T_{k-1}$ --- this is because $\shrink_{k+1}$ cannot create a block of length at least three consisting of at least two distinct symbols.
  Symmetrically, if the rightmost block of $\hX_{k-1}$ consists of two distinct symbols, then there is a block boundary at the end of every occurrence of $\hX_{k-1}$ in $T_{k-1}$.
  In all cases, the blocks collapsed to $\hX_{k}$ are processed in the same way and induce an occurrence of $\hX_{k}$ in $T_{k}$.
  The starting position $i_{k}$ of this occurrence satisfies $|\exp(T_{k}\fragmentco{0}{i_{k}})|=|\exp(T_{k-1}\fragmentco{0}{i_{k-1}}\cdot L_{k-1})|=|\exp(T_{k-1}\fragmentco{0}{i_{k-1}})|+ |\exp(L_{k-1})|$, which is $i+|\exp(L_0\cdots L_{k-2})|+|\exp(L_{k-1})|=i+|\exp(L_0\cdots L_{k-1})|$ by the inductive hypothesis.
\end{proof}

Next, we address the task of efficiently computing a popped sequence.
\begin{lemma}\label{lem:computepseq}
  Let $\G$ be an $r$-round restricted recompression RLSLP representing a text $T$.
  Given a fragment $X=T\fragmentcc{i}{j}$, the run-length encoding of the popped sequence $\PSeq(X)$ can be computed in $\Oh(r)$ time,
  along with a decomposition $\PSeq(X)=L_0\cdots L_q\cdot R_q\cdots R_0$ for $q=\max\{k\in \Zz : \hX_k\ne \emptystring\}$.
\end{lemma}
\begin{proof}
For each round $k$ such that $\hX_k\ne \emptystring$, our algorithm assumes to be given (pointers to) nodes $\unu_k,\umu_k\in \uTr$ representing, respectively, the leftmost and the rightmost symbol within the occurrence of $\hX_k$ in $T_k$ stipulated by Lemma~\ref{lem:occx}.
For $k=0$, the nodes $\unu_0$ and $\umu_0$ correspond to $T\position{i}$ and $T\position{j}$, respectively, and they can be obtained in $\Oh(r)$ time by Observation~\ref{obs:leaf}.

For subsequent levels $k\in \Zz$, the algorithm first considers the parents $\unu'_k\coloneq \parent{\unu_k}$ and $\umu'_k\coloneq \parent{\umu_k}$.
The nodes $\unu'_k$ and $\umu'_k$ correspond to, respectively, the leftmost and the rightmost block of $\hX_{k}$ with respect to $\shrink_{k+1}$, possibly including some extra symbols to the left or to the right of $\hX_k$, respectively.
In particular, since we already assumed that $\hX_k\ne \emptystring$, then $L_k=\emptystring$ if and only if $\unu'_k$ has two children with distinct symbols, $\unu_k$ is the left child of $\unu'_k$, and $\unu_k\ne \umu_k$ (if $\unu_k=\umu_k$, then the only block of $\hX_k$ consists of a single symbol).
Similarly, $R_k=\emptystring$ if (but not only if) $\umu_k$ has two children with distinct symbols and $\umu_k$ is the right child of $\umu'_k$.
The other possibility for $R_k=\emptystring$ is when $L_k\ne \emptystring$ and $\unu'_k=\umu'_k$ (that is, $\hX_k$ consists of a single block that does not have two distinct symbols).
In this special case,  $L_k=\hX_k$ is the power of $\symb{\unu_k}$ with exponent equal to the number of siblings between $\unu_k$ and $\umu_k$ (inclusive).

Except for the special case, we have $\unu_{k+1}=\unu'_k$ if $L_k = \emptystring$.
Otherwise, $\unu_{k+1}=\nxt{\unu'_k}$ and, by \cref{obs:one_run}, $L_k$ is a power of $\symb{\unu_k}$ with exponent equal to the number of right siblings of $\unu_k$ (including $\unu_k$).
Symmetrically, $\umu_{k+1}=\umu'_k$ if $R_k = \emptystring$.
Otherwise, $\umu_{k+1}=\prv{\umu'_k}$ and, by \cref{obs:one_run},  $R_k$ is a power of $\symb{\umu_k}$ with exponent equal to the number of left siblings of $\umu_k$ (including $\umu_k$).

Before we proceed to the next level, we must stop when $\hX_{k+1}= \emptystring$; this condition holds when $L_k\ne \emptystring$ and $\unu'_k = \umu'_k$ (the special case considered above) as well as when $L_k\ne \emptystring$, $R_k\ne \emptystring$, and $\nxt{\unu'_k}=\umu'_k$.

The sequence of $\parent{\cdot}$ and $\nxt{\cdot}$ operations involving nodes $\unu_k$ and $\unu'_k$, as well as the sequence of $\parent{\cdot}$ and $\prv{\cdot}$ operations involving nodes $\umu_k$ and $\umu'_k$ satisfy the conditions of \cref{lem:traverse}, so the overall running time is $\Oh(q+r)=\Oh(r)$.
\end{proof}

\subsection{LCE Queries in $\Oh(r)$ time}\label{sec:lce}
Recall that an LCE query, given positions $i$ and $i'$ in $T\in \Sigma^n$, returns
  \[\LCE_T(i,i')=\max\{d\in \fragmentcc{0}{n-\max(i,i')} :T\fragmentco{i}{i+d}=T\fragmentco{i'}{i'+d}\}.\]

As a sample application of our implementation of parse tree nodes and our variant of popped sequences, we show how to obtain $\Oh(r)$-time LCE queries in our setting.
Our recursive procedure answering LCE queries is an easy to implement version of the algorithm originating from~\cite[Theorem 5.25]{KK23} and analyzed using popped sequences in the preliminary version of this paper~\cite{DK24}.

\begin{proposition}\label{prp:lce}
  LCE queries on a text $T$  represented using an $r$-round restricted recompression run-length straight-line program can be answered in $\Oh(r)$ time.
  \end{proposition}
  \begin{proof}
    \newcommand{\lce}[2]{\mathtt{lce}(#1,#2)}

    We implement a recursive procedure $\lce{\nu}{\nu'}$ that, given (pointers to) nodes $\nu,\nu'\in \Tr$, returns $\LCE_T(\pos{\nu},\pos{\nu'})$.
    To streamline the presentation, we also assume $\lce{\nu}{\nil}=\lce{\nil}{\nu'}=0$.

    Our implementation relies on a helper function $\nxt{\nu}$ that returns the highest node $\mu$ such that $\nexp{\mu}$ immediately follows $\nexp{\nu}$, that is, $\pos{\mu}=\pos{\nu}+\len{\symb{\nu}}$.
    Formally,
    \begin{itemize}
      \item If $\parent{\nu}=\nil$, then $\nxt{\nu}=\nil$, i.e., $\nxt{\nu}$ is undefined if $\nu$ is the root of $\Tr$.
      \item If $\sibling{\nu}{1}\ne \nil$, then $\nxt{\nu}=\sibling{\nu}{1}$, i.e., $\nxt{\nu}$ is the right sibling of $\nu$ if such a sibling exists.
      \item Otherwise, $\nxt{\nu}=\nxt{\parent{\nu}}$. In this case, $\nu$ does not have any right sibling so $\pos{\nu}+\len{\symb{\nu}}=\pos{\parent{\nu}}+\len{\symb{\parent{\nu}}}$ and the value $\nxt{\nu}$ can be computed recursively.
    \end{itemize}

    We can now implement $\lce{\nu}{\nu'}$ assuming $\nu\ne \nil$ and $\nu'\ne \nil$.
    \begin{enumerate}
      \item\label{it:lce:1} If $\len{\symb{\nu}} = \len{\symb{\nu'}} = 1$ and $\symb{\nu}\ne \symb{\nu'}$, then we return $0$.
      \item\label{it:lce:2} If $\len{\symb{\nu}} = \len{\symb{\nu'}} > 1$ and $\symb{\nu}\ne \symb{\nu'}$, then we return $\lce{\child{\nu}{0}}{\child{\nu'}{0}}$.
      \item\label{it:lce:3} If $\len{\symb{\nu}} > \len{\symb{\nu'}}$, then we return $\lce{\child{\nu}{0}}{\nu'}$.
      \item\label{it:lce:4} If $\len{\symb{\nu}} < \len{\symb{\nu'}}$, then we return $\lce{\nu}{\child{\nu'}{0}}$.
      \item\label{it:lce:5} If $\symb{\nu}=\symb{\nu'}$ and $\sibling{\nu}{1}=\nil$ or $\sibling{\nu'}{1}=\nil$, then we return $\len{\symb{\nu}}+\lce{\nxt{\nu}}{\nxt{\nu'}}$.
      \item\label{it:lce:6} Otherwise, $\symb{\nu}=\symb{\nu'}$ and $\sibling{\nu}{1}\ne\nil \ne \sibling{\nu}{1}$.
      We pick the largest $d$ such that $\sibling{\nu}{d}\ne\nil \ne \sibling{\nu}{d}$ and return $d\cdot \len{\symb{\nu}} + \lce{\sibling{\nu}{d}}{\sibling{\nu'}{d}}$.
    \end{enumerate}

  Let us justify the correctness of the algorithm.
  In case~\ref{it:lce:1}, we have $\T\position{\pos{\nu}}=\symb{\nu}\ne \symb{\nu'}=T\position{\pos{\nu'}}$, so the algorithm correctly returns $\LCE_T(\pos{\nu},\pos{\nu'})=0$.
  In the remaining cases, we make a recursive call, which we denote $\lce{\mu}{\mu'}$.
  In cases~\ref{it:lce:2}--\ref{it:lce:4}, we have $\pos{\mu}=\pos{\nu}$ and $\pos{\mu'}=\pos{\nu'}$, so the value $\LCE_T(\pos{\nu},\pos{\nu'})=\LCE_T(\pos{\mu},\pos{\mu'})$ is computed correctly.
  Moreover, whenever $\mu=\child{\nu}{0}$ or $\mu'=\child{\nu'}{0}$, we have $\len{\symb{\nu}}>1$ or $\len{\symb{\nu'}}>1$, respectively, so the leftmost children exist.
  In the remaining cases, we have $\symb{\nu}=\symb{\nu'}$ and thus $\LCE_T(\pos{\nu},\pos{\nu'})=\len{\symb{\nu}}+\LCE_T(\pos{\nu}+\len{\symb{\nu}},\pos{\nu'}+\len{\symb{\nu}})$.
  In case~\ref{it:lce:5}, we specifically pick $\mu=\nxt{\nu}$ and $\mu'=\nxt{\nu'}$ so that $\pos{\mu}=\pos{\nu}+\len{\symb{\nu}}$ and $\pos{\mu'}=\pos{\nu'}+\len{\symb{\nu'}}$, and hence the returned value is correct. (This argument works even if $\mu=\nil$ or $\mu'=\nil$ as long as we assume $\pos{\nil}=|T|$.)
  The same is true in case~\ref{it:lce:6} provided that $d=1$.
  If $d\ge 2$, on the other hand, then the parents of $\nu$ and $\nu'$ have at least three children each, so the underlying symbols must be runs.
  In this case, all the siblings of $\nu$ and $\nu'$ have the same symbol $\symb{\nu}$.
  Consequently, $\LCE_T(\pos{\nu},\pos{\nu'})=d\cdot \len{\symb{\nu}}+\LCE_T(\pos{\nu}+d\cdot \len{\symb{\nu}},\pos{\nu'}+d\cdot \len{\symb{\nu}})=d\cdot \len{\symb{\nu}}+\LCE_T(\pos{\mu},\pos{\mu'})$, and hence the returned value is correct.
  Finally, observe that the algorithm terminates because, throughout the recursion, both $\nu$ and $\nu'$ move forward with respect to the pre-order traversal of $\Tr$.

  In order to compute $\ell\coloneq \LCE_T(i,i')$, we run $\lce{\nu}{\nu'}$, where $\nu$ and $\nu'$ are the highest nodes such that $\pos{\nu}=i$ and $\pos{\nu'}=i'$, respectively.
  Such nodes can be obtained in $\Oh(r)$ time by initializing $\nu$ and $\nu'$ to be the root of $\Tr$ and repeatedly setting $\child{\nu}{\idx{\nu}{i}}$ or $\child{\nu'}{\idx{\nu'}{i'}}$ as long as $\pos{\nu}< i$ or $\pos{\nu'}< i'$, respectively.
  It remains to analyze the time of evaluating $\LCE_T(i,i')$.
  For this, we say that nodes $(\nu,\nu')$ form a \emph{matching pair}
  if  $\symb{\nu}=\symb{\nu'}$ and $T\fragmentco{i}{\pos{\nu}}=T\fragmentco{i'}{\pos{\nu'}}$.
  Denote by $N$ and $N'$ the set of nodes $\nu$ and $\nu'$ participating in matching pairs; these pairs form a perfect matching between $N$ and $N'$.
  We claim that the query algorithm satisfies the following additional invariant:
  No proper ancestor of $\nu$ belongs to $N$ and no proper ancestor of $\nu'$ belongs to $N'$.
  The invariant is satisfied at the beginning because all proper ancestors of $\nu$ and $\nu'$ have their expansions starting before position $i$ and $i'$, respectively.
  In cases~\ref{it:lce:2} and~\ref{it:lce:4}, the node $\nu$ does not form a matching pair with any proper ancestor of $\nu'$ (they do not belong to $N'$ by our invariant),
  with $\nu'$ itself (because $\symb{\nu}\ne \symb{\nu'}$), nor with any proper descendant of $\nu'$ (their symbols are shorter than $\len{\symb{\nu'}}\le \len{\symb{\nu}}$). Thus, $\nu$ does not belong to $N$.
  By symmetry, the invariant remains satisfied in cases~\ref{it:lce:2} and~\ref{it:lce:3} when we replace $\nu'$ by its leftmost child.
  In cases~\ref{it:lce:5} and~\ref{it:lce:6}, proper ancestors of the nodes $\mu$ and $\mu'$ are proper ancestors of the nodes $\nu$ and~$\nu'$.

  Let $\hat{N}$ consist of nodes that do not belong to $N$ yet have descendants in $N$, and define $\hat{N}'$ analogously based on $N'$ instead of $N$.
  As we trace the value $\nu$ in subsequent recursive calls (including the evaluation of the $\nxt{\cdot}$ function), we obtain an Euler tour traversal of a subtree of $\Tr$ consisting of selected nodes in $\hat{N}$, their children, and possible the ancestors of $\leaf{i+\ell}$ (if $\leaf{i+\ell}\ne \nil$)
  Similarly, $\nu'$ visits selected nodes in $\hat{N}'$, their children, and possibly the ancestors of $\leaf{i'+\ell}$ (if $\leaf{i'+\ell}\ne \nil$).
  Moreover, Case~\ref{it:lce:6} guarantees that, for any node in $\hat{N}$ or $\hat{N}'$, we visit at most two among its children in $N$ or $N'$, respectively.
  Thus, the total runtime of the algorithm is $\Oh(|\hat{N}|+|\hat{N}'|+r)$.

  It remains to prove that $|\hat{N}|=\Oh(r)$ and, by symmetry, $|\hat{N}'|=\Oh(r)$.
  For this, observe that Lemma~\ref{lem:occx} applied to $T\fragmentco{i}{i+\ell}=T\fragmentco{i'}{i'+\ell}$ implies that the nodes participating in the popped sequence $\PSeq(T\fragmentco{i}{i+\ell})$ belong to $N$, so $\hat{N}$ may only contain proper ancestors of these nodes.
  By Definition~\ref{def:popped}, the nodes in the popped sequence have $\Oh(r)$ parents in total (since nodes within each $L_k$ and $R_k$ have a single parent). These nodes also have $\Oh(r)$ proper ancestors because the parse tree is of height $\Oh(r)$ and does not contain degree-one vertices.
  \end{proof}

  We remark that a representation of $T$ as an $r$-round restricted recompression RLSLP constitutes an analogous representation of the reverse text $\overline{T}$ (it suffices to reverse the right-hand side of every production), so
  Proposition~\ref{prp:lce} can also be used to answer the following queries $\overline{\LCE}_T$ queries:
  \[\overline{\LCE}_T(i,i') := \max\{ d\in\fragmentcc{0}{\min (i,i')} : T\fragmentco{i-d}{i}= T\fragmentco{i'-d}{i'}\}.\]

\section{IPM query algorithm: Overview}\label{sec:algorithm}
Recall that the query algorithm is given two fragments $X$ and $Y$ of the text $T$ satisfying $|Y|< 2|X|$,
and the goal is to identify exact occurrences of $X$ contained within $Y$.
Our $\Oh(r)$-time query algorithm performs the following steps:
\begin{enumerate}
\item Define $\ell=\max\{k \in \Zz: |\hX_k|>k\}$, where $(\hX_k)_{k=0}^\infty$ is introduced in Definition~\ref{def:popped} along with the popped sequence $\PSeq(X)$. As proved in Section~\ref{sec:level}, $\rle(\hX_\ell)$ can be constructed in $\Oh(r)$ time.
\item Define an appropriate fragment $Y'_\ell$ of $T_\ell$ such that every occurrence of $X$ within $Y$ yields an occurrence of $\hX_\ell$ within $Y'_\ell$. As proved in Section~\ref{sec:find}, we can ensure that $|Y'_\ell|=\Oh(|\hX_\ell|)$ and $\rle(Y'_\ell)$ can be constructed in $\Oh(r)$ time.
\item Use pattern matching in run-length encoded strings to find the occurrences of $\hX_\ell$ in $Y'_\ell$, represented as $\Oh(1)$ arithmetic progressions; see Section~\ref{sec:pm}.
\item For each progression, use $\Oh(1)$ LCE queries (in $\Oh(r)$ time each) to find the occurrences of $\hX_\ell$ in $Y'_\ell$ that extend to occurrences of $X$ in $Y$; see Section~\ref{sec:verify}.
\end{enumerate}

\section{\texorpdfstring{\boldmath Constructing the proxy pattern \(\hX_\ell\)}{Constructing the proxy pattern X̄ₗ}}\label{sec:level}
Observe that, for every non-empty fragment $X$ of $T$, we have $|\hX_0|=|X|>0$ and $|\hX_k|=0$ for $k\in \mathbb{Z}_{>r}$, so $\ell:=\max\{k \in \Zz : |\hX_k|>k\}$ is well-defined.

\begin{lemma}\label{lem:hx}
  The level $\ell:=\max\{k \in \Zz : |\hX_k|>k\}$ and the run-length encoding $\rle(\hX_\ell)$ of the proxy pattern can be constructed in $\Oh(r)$ time.
\end{lemma}
\begin{proof}
  Let us first focus on computing the level $\ell$. For this, we use \cref{lem:computepseq} to compute $\rle(\PSeq(X))$ decomposed as $\rle(L_0)\cdots \rle(L_q)\cdot \rle(R_q)\cdots \rle(R_0)$, where $q=\max\{k\in \Zz : \hX_k \ne \emptystring\}$.

  For $k\in \Zz$, define $M_k$ to be the \emph{multiset} of symbols appearing in $\hX_k$, that is, $\{\hX_k\position{i} : i\in \fragmentco{0}{|\hX_k|}\}$.
  Our strategy to iterate over levels $k\in \fragmentcc{0}{q}$ in the decreasing order maintaining a multiset $M$ such that $M=M_k$ holds as soon as we complete processing level $k$.
  Initially, we set $M\coloneqq \emptyset = M_{q+1}$.
  Thus, for each level $k$, our goal is to transform $M_{k+1}$ into $M_k$.
  For this, we need to insert symbols present in $L_k$ and $R_k$.
  Additionally, we need to expand symbols of $\hX_{k+1}$ created by $\shrink_{k+1}$.
  Thus, for each symbol $A\in M$ with $\lvl(A)=k+1$, we replace $A$ with symbols of $\rhs(A)$.

  In order to efficiently implement these operations, we store $M$ as a monotone bucket priority queue~\cite{MS08}, where the priority of each symbol $A$ is its level $\lvl(A)$.
  In other words, we maintain an array $Q\fragmentcc{0}{q}$, where $Q\position{k'}$ is a list of symbols $A\in M$ such that $\lvl(A)=k'$ (each with the same multiplicity as in $M$).
  Upon inserting an element to $A$ to $M$, we simply append it to the list $Q\position{\lvl(A)}$.
  Moreover, in order to retrieve all $A\in M$ with $\lvl(A)=k+1$, we iterate over the list $Q\position{k+1}$.
  For each such symbol $A$, our algorithm removes $A$ from $M$ and inserts all symbols of $\rhs(A)$ into $M$.

  Recall that our goal is to identify the largest $k$ such that $|\hX_{k}|>k$ or, equivalently, $|M_k|>k$.
  Furthermore, observe that every operation on the set $M$ increases the size $|M|$ --- whenever we remove $A\in M$, we immediately insert $|\rhs(A)|\ge 2$ symbols of $\rhs(A)$.
  Consequently, it suffices to process subsequent integers $k$ from $q$ down to $0$ and, while transforming $M$ from $M_{k+1}$ to $M_k$, stop as soon as $|M|>k$; we are then guaranteed that $|M_k|\ge |M|>k$.
  Since the stopping condition has not been satisfied at higher levels, we can also conclude that $\ell = k$ holds in this case.

  As far as the running time of the algorithm is concerned, note that constructing $\rle(\PSeq(X))$ takes $\Oh(r)$ time.
  Initializing the array $Q$ takes $\Oh(q)=\Oh(r)$ time.
  Finally, note that each operation on $M$ takes constant time and the number of operations on $M$ does not exceed $3|M|$ (because every deletion is immediately followed by at least two insertions).
  Consequently, since $|M|\le \ell+\Oh(1)=\Oh(r)$ holds when the algorithm terminates, we conclude that the main phase of the algorithm also takes $\Oh(r)$ time.

  In the next step, we construct $\hX_{\ell+1}$.
  For this, it suffices to start with $L_{\ell+1}\cdots L_q\cdot R_q\cdots R_{\ell+1}$ and exhaustively expand every symbol at a level higher than $\ell+1$.
  Our implementation processes subsequent symbols in $L_{\ell+1}\cdots L_q\cdot R_q\cdots R_{\ell+1}$ from left to right using a recursive procedure that, for a given symbol~$A$, either outputs $A$ (if $\lvl(A)\le \ell+1$) or recursively processes symbols of $\rhs(A)$ in the left-to-right order (otherwise).
  Since each production consists of at least two symbols, the total running time of this step is $\Oh(q-\ell+|\hX_{\ell+1}|)=\Oh(r)$.

  Finally, in order to determine $\rle(\hX_\ell)$, we expand every symbol $A$ in $\hX_{\ell+1}$ with $\lvl(A)=\ell+1$ into $\rle(\rhs(A))$.
  We then prepend $\rle(L_\ell)$, append $\rle(R_\ell)$, and post-process the resulting sequence so that two powers of the same symbol are merged whenever they appear next to each other.
  This step takes $\Oh(1+|\hX_{\ell+1}|)=\Oh(r)$ time because $|\rle(\rhs(A))|\le 2$ holds for every symbol $A$.
\end{proof}

\section{\texorpdfstring{\boldmath Constructing the proxy text $Y'_\ell$}{Constructing the proxy text Yₗ'}}\label{sec:find}
The following Lemma~\ref{lem:defyp} defines the appropriate proxy text $Y'_\ell$ as a fragment of $T_\ell$ such that $\exp(Y'_\ell)$ covers a sufficiently large fragment of $Y$.
A natural choice would be to pick $Y'_\ell$ so that $\exp(Y'_\ell)$ covers the entire $Y$; unfortunately, this may result in a proxy text $Y'_\ell$ that is much longer than the proxy pattern $X'_\ell$, even with respect to run-length encoding: the fact that $X$ is covered by a length-$\Oh(\ell)$ fragment of $T_{\ell+1}$ does not imply that the same is true for $Y$.
Nevertheless, if $X$ occurs in $Y$, the fragment of $Y$ containing all the occurrences of $X$ is guaranteed to be covered by a length-$\Oh(\ell)$ fragment of $T_{\ell+1}$.
Our definition assures that $Y'_\ell$ is long enough so that every occurrence of $X$ in $Y$ yields an occurrence of $\hX_{\ell}$ in $Y'_\ell$ and, at the same time, short enough to enable efficient implementation of our query algorithm.
Specifically, we define $Y'_\ell$ so that it covers the middle position $Y$ (which is contained in every occurrence of $X$ in $Y$) and extends in both directions just enough to cover, for every occurrence $X$ containing that middle position, the induced occurrence of $\hX_\ell$ in $T_\ell$.
We quantify this extension using the length (so that the occurrences of $\hX_\ell$ in $Y'_\ell$ form $\Oh(1)$ arithmetic progressions; see Section~\ref{sec:pm}) and the number of runs (so that $\rle(Y'_\ell)$ can be constructed efficiently; see Lemma~\ref{lem:yp}).
We also ensure that $\exp(Y'_\ell)$ does not extend beyond $Y$.

\begin{lemma}\label{lem:defyp}
  Let $T\position{m}$ be the middle character in \(Y\) and let $T_\ell\position{m_\ell}$  be its ancestor in $T_{\ell}$.
  Moreover, let $e\in \fragmentco{m_\ell}{|T_\ell|}$ be the largest position such that $|T_\ell\fragmentcc{m_\ell}{e}|\le |\hX_\ell|+\ell$ and $|\shrink_{\ell+1}(T_\ell\fragmentcc{m_\ell}{e})|\le 2\ell+3$; symmetrically, let $b\in \fragmentcc{0}{m_\ell}$ be the smallest position such that $|T_\ell\fragmentcc{b}{m_\ell}|\le |\hX_\ell|+\ell$ and $|\shrink_{\ell+1}(T_\ell\fragmentcc{b}{m_\ell})|\le 2\ell+3$.
  Define the \emph{proxy text} $Y'_\ell$ as the maximal fragment of $T_\ell\fragmentcc{b}{e}$ whose expansion \(\exp(Y'_\ell)\) is a fragment of $T$ contained within \(Y\).\footnote{The proxy text is empty if there is no character in $T_\ell\fragmentcc{b}{e}$ whose expansion is contained within $Y$.}

  For every occurrence of $X$ contained in $Y$, the corresponding occurrence of $\hX_\ell$ in $T_\ell$ stipulated by  \cref{lem:occx} is contained within $Y'_\ell$.
\end{lemma}
\begin{proof}
  Let us fix an occurrence $T\fragmentcc{i}{j}$ of $X$ contained in $Y$.
  For every $k\in \fragmentcc{0}{\ell}$, \cref{lem:occx} yields a corresponding occurrence $T_k\fragmentcc{i_k}{j_k}=\hX_k$.
  Let us also define $T_k\position{i'_k}$ as the ancestor of $T\position{i}$ in $T_k$.

  \begin{claim}\label{clm:defyp}
  For every $k\in \fragmentcc{0}{\ell}$, we have $|T_k\fragmentco{i'_k}{i_k}|\le k$.
  \end{claim}
  \begin{claimproof}
  We proceed by induction. The claim is trivial for $k=0$ due to $i'_0=i_0=i$.
  For $k>0$, \cref{lem:occx} implies $T_k\fragmentco{0}{i_k}=\shrink_{k}(T_{k-1}\fragmentco{0}{i_{k-1}} \cdot L_{k-1})$.
  In particular, $T_k\fragmentco{i'_k}{i_k}=\shrink_{k}(L'_{k-1}\cdot T_{k-1}\fragmentco{i'_{k-1}}{i_{k-1}}\cdot L_{k-1})$, where $L'_{k-1}$ consists of the children of $T_k\position{i'_k}$ to the left of $T_{k-1}\position{i'_{k-1}}$.
  Observe that $\shrink_k$ does not place block boundaries within $L'_{k-1}\cdot T_{k-1}\position{i'_{k-1}}$ (all those symbols have $T_k\position{i'_k}$ as their parent) and within $L_{k-1}$ (by \cref{def:popped}, $L_{k-1}$ is either empty or the leftmost block of $\hX_{k-1}$).
  The remaining possible locations for block boundaries within $L'_{k-1}\cdot T_{k-1}\fragmentco{i'_{k-1}}{i_{k-1}}\cdot L_{k-1}$ are immediately after the characters of $T_{k-1}\fragmentco{i'_{k-1}}{i_{k-1}}$,
  and the number of such locations does not exceed $|T_{k-1}\fragmentco{i'_{k-1}}{i_{k-1}}|\le k-1$ by the inductive hypothesis.
  Hence, the number of blocks that $L'_{k-1}\cdot T_{k-1}\fragmentco{i'_{k-1}}{i_{k-1}}\cdot L_{k-1}$ is partitioned into (with respect to $\shrink_k$) equals $|T_{k}\fragmentco{i'_{k}}{i_k}| \le k$.
  \end{claimproof}

  The condition $|Y|<2|X|$ implies that $T\fragmentcc{i}{j}$ contains the middle position of $Y$.
  Thus, $i \le m$ and $i'_\ell \le m_\ell$.
  Consequently, \cref{clm:defyp} implies $i_\ell \le i'_\ell + \ell \le m_\ell+\ell$.
  Hence, we have $j_\ell = i_\ell + |\hX_{\ell}|-1 \le m_\ell+|\hX_{\ell}|+\ell-1$ and $|T_\ell\fragmentcc{m_\ell}{j_\ell}| \le |\hX_{\ell}|+\ell$.
  Therefore, $T_\ell\fragmentcc{m_\ell}{j_\ell}$ is either a suffix of $\hX_{\ell}$ or consists of the entire $\hX_{\ell}$ preceded by at most $\ell$ symbols.
  Since $|\shrink_{\ell+1}(\hX_\ell)|=|\shrink_{\ell+1}(L_\ell)\cdot \hX_{\ell+1}\cdot \shrink_{\ell+1}(R_\ell)|\le |\hX_{\ell+1}|+2 \le \ell+3$, we conclude that $|\shrink_{\ell+1}(T_\ell\fragmentcc{m_\ell}{j_\ell})| \le \ell+3 + \ell =2\ell+3$.

  Recall that $e\in \fragmentcc{m_\ell}{|T_\ell|}$ has been chosen as the largest index such that $|T_\ell\fragmentcc{m_\ell}{e}|\le |\hX_{\ell}|+\ell$ and $|\shrink_{\ell+1}(T\fragmentcc{m_\ell}{e})|\le 2\ell+3$.
  As proved above, the index $j_\ell$ satisfies both conditions, so we derive $j_\ell \le e$.
  A symmetric argument shows that $b \le i_\ell$, and thus $T_\ell\fragmentcc{i_\ell}{j_\ell}$ is contained within $T_\ell\fragmentcc{b}{e}$.

  By \cref{def:popped}, the expansion $\exp(T_\ell\fragmentcc{i_\ell}{j_\ell})$ is contained within $T\fragmentcc{i}{j}$ and, by our assumption, $T\fragmentcc{i}{j}$ is contained within $Y$, so the expansion $\exp(T_\ell\fragmentcc{i_\ell}{j_\ell})$ is contained within $Y$.
  By construction, $Y'_\ell$ is the maximal fragment contained within $T_\ell\fragmentcc{b}{e}$ whose expansion $\exp(Y'_\ell)$ is contained within $Y$.
  Hence, $T_\ell\fragmentcc{i_\ell}{j_\ell}$ must be contained within $Y'_\ell$.
\end{proof}

\begin{lemma}\label{lem:yp}
  The run-length encoding $\rle(Y'_\ell)$ of the proxy text $Y'_\ell$ can be constructed in \(\bigO(r)\) time.
\end{lemma}
\begin{proof}
  First, we compute a pointer to $T\position{m}$ using Observation~\ref{obs:leaf},
  and then we move up until we arrive at a node of $\uTr$ representing a character $T_{\ell+1}\position{m_{\ell+1}}$ at level $\ell+1$.
  Next, we generate the following fragment of $T_{\ell+1}$: \[\hat{Y}_{\ell+1}:=T_{\ell+1}\fragmentco{\max(0,m_{\ell+1}-2\ell-2)}{\min(|T_{\ell+1}|,m_{\ell+1}+2\ell+1)}.\]
  To do this, move from $T_{\ell+1}\position{m_{\ell+1}}$ by $2\ell+2$ positions in both directions.
  This phase is implemented in $\Oh(r+\ell+1)=\Oh(r)$ time using Lemma~\ref{lem:traverse} (to move to the right) and its symmetric counterpart (to move to the left).

  To derive $\rle(Y'_\ell)$ from $\hat{Y}_{\ell+1}$, we replace $A$ with $\rle(\rhs(A))$ for every symbol $A$ with $\lvl(A)=\ell+1$, and then we trim the resulting string to make sure that it contains at most $|\hX_{\ell}|+\ell-1$ symbols on either side of $T_\ell\position{m_\ell}$ and that its expansion does not exceed beyond $Y$.
  This phase costs $\Oh(\ell+1)$ time, for a total of $\Oh(r)$ time for the entire algorithm.
\end{proof}

\section{\texorpdfstring{\boldmath Finding the occurrences of $\hX_\ell$ in $Y'_\ell$}{Finding the occurrences of X̄ₗ in Yₗ'}}\label{sec:pm}
As observed in \cite{AB92,ALV92,Chu95}, exact pattern matching in run-length encoded strings can be implemented in linear time.
Below, we adapt this result to support representing the occurrences as arithmetic progressions.

\begin{proposition}[see {\cite[Theorem 1]{Chu95}}]\label{prp:rlepm}
  Given the run-length encodings $\rle(P)$, $\rle(S)$ of strings $P,S\in \Sigma^+$, the set $\Occ(P,S)=\{i\in \fragmentcc{0}{|S|-|P|}: P=S\fragmentco{i}{i+|P|}\}$, represented by at most $\min(|\rle(S)|,\lfloor|S|/|P|\rfloor)$ arithmetic progressions with difference at most $|P|$, can be constructed in $\Oh(|\rle(P)|+|\rle(S)|)$~time.
\end{proposition}
\begin{proof}
  None of the papers \cite{AB92,ALV92,Chu95} considers the case of $|\rle(P)|=1$, in which reporting $\Occ(P,S)$ explicitly might be too costly.
  In this case, every occurrence of $P$ is contained within a single run of $S$.
  A run $S\fragmentco{j}{j+q}=b^q$ contains an occurrence of $P=a^p$ if and only if $a=b$ and $p\le q$; then, the occurrences form an arithmetic progression $(j,j+1,\ldots,\allowbreak j+q-p)$ with difference one.
  Each such progression corresponds to a run of length at least $|P|$ in $S$, so their number does not exceed $\min(|\rle(S)|,\lfloor|S|/|P|\rfloor)$.
  To construct them in $\Oh(p+s)$ time, we scan $\rle(S)$ maintaining the length of the already processed prefix~of~$S$.

  For $|\rle(P)|\ge 2$, the algorithm of \cite[Theorem 1]{Chu95} reports elements of $\Occ(P,S)$ one be one; the number of occurrences does not exceed $|\rle(S)|-1$.
  By \cite[Fact 1.1]{KRRW15}, if we greedily group occurrences into arithmetic progressions with a difference at most $|P|$ each, we obtain at most $\lfloor{|S|/|P|}\rfloor$ progressions in total.
\end{proof}

Applying Proposition~\ref{prp:rlepm} to $\hX_\ell$ and $Y'_\ell$, we obtain the following result:

\begin{corollary}\label{cor:hxinyp}
  Given $\rle(\hX_\ell)$ and $\rle(Y'_\ell)$, the occurrences of $\hX_\ell$ in $Y'_\ell$, represented as $\Oh(1)$ arithmetic progressions, each with a difference of at most $|\hX_\ell|$, can be computed in $\Oh(r)$ time.
\end{corollary}
\begin{proof}
By Lemmas~\ref{lem:hx} and~\ref{lem:yp}, both $\rle(\hX_\ell)$ and $\rle(Y'_\ell)$ are of size $\Oh(r)$.
Moreover, the definition of $Y'_\ell$ (see Lemma~\ref{lem:defyp}) guarantees that $|Y'_\ell|< 2|\hX_\ell|+2\ell < 4|\hX_\ell|$, so the algorithm of Proposition~\ref{prp:rlepm} returns at most four progressions.
\end{proof}

\newcommand{\oa}{\overline{a}}
\newcommand{\oc}{\overline{c}}
\newcommand{\ou}{\overline{u}}
\newcommand{\ov}{\overline{v}}

\section{Verifying candidate occurrences}\label{sec:verify}
  Lemma~\ref{lem:defyp} shows that every occurrence of \(X\) in \(Y\) corresponds to an occurrence of $\hX_\ell$ in $Y'_\ell$.
  Moreover, by Corollary~\ref{cor:hxinyp}, \( \Occ(\hX_\ell, Y'_\ell)\) is the union of \(\bigO(1)\) arithmetic progressions, each with a difference of at most \(|\hX_\ell|\).

  Let \(V_\ell := \{ \oa_\ell+ i\cdot g_\ell : i\in \fragmentco{0}{s_\ell}\}\) be one such arithmetic progression.
  Observe that $Y'_{\ell}\fragmentco{\oa_\ell+(i-1)\cdot g_\ell}{\oa_\ell+i\cdot g_\ell}=\hX_{\ell}\fragmentco{0}{g_\ell}$ holds for each $i\in \fragmentco{1}{s_\ell}$.
  Thus, the arithmetic progression $V_\ell$ of positions in $Y'_\ell$ corresponds to an arithmetic progression
  $V=\{\oa + i\cdot g : i\in \fragmentco{0}{s}\}$ of positions in $Y$, where $s=s_\ell$, \(g = |\exp(\hX_\ell\fragmentco{0}{g_\ell})|\), and \(Y\position{\oa}\) is the leftmost leaf in the subtree of $\uTr$ rooted at $Y'_\ell\position{\oa_\ell}$.
  Since the non-terminals store their expansion lengths, a single pass over $\rle(\hX_{\ell})$ and \(\rle(Y'_\ell)\) lets us compute the values $g$ and $\oa$, respectively.

  Consequently, we henceforth assume that each arithmetic progression $V_\ell\subseteq  \Occ(\hX_\ell, Y'_\ell)$ is already represented by the underlying arithmetic progression $V$ of positions in~$Y$; note that each such progression $V$ consists of starting positions of occurrences of $\exp(\hX_\ell)$ in $Y$, and its difference $g$ satisfies $g\le |\exp(\hX_\ell)|$.

  Define $\oc=|\exp(L_0\cdots L_{\ell-1})|$ and $c=|X|-|\exp(R_{\ell-1}\cdots R_0)|$ so that $\exp(\hX_{\ell})=X\fragmentco{\oc}{c}$.
  Let \(V = \{ \oa + i\cdot g : i\in\fragmentco{0}{s}\}\) be any of the aforementioned arithmetic progressions representing occurrences of $\exp(\hX_\ell)$ in $Y$.
  Our goal is to test, in bulk, for each position $\oa+i\cdot g \in V$, whether $X$ occurs in $Y$ at position $\oa+i\cdot g -\oc$.
  Then, the only remaining step is to filter out occurrences of $X$ that are not contained within $Y$.
  If $|V|=1$, one can use a single LCE query to verify $Y\fragmentco{\oa-\oc}{\oa-\oc+|X|}=X$.
  Thus, we henceforth assume $|V|>1$.

   \begin{figure}[t]
  \centering

  \begin{minipage}{0.45\textwidth}
\centering
  \begin{tikzpicture}[scale=0.5]
    \draw[dashed] (0,0) -- (2,3);
    \draw[dashed] (8,3) -- (10.5,0);
    \draw (2,3) rectangle (8,4);
    \node at (5,3.5) {\(\hX_\ell\)};

    \node[blue] at (4.5,0.3) {\(X\)};
    \draw[blue, thick] (0,0) -- (10.5,0);

    \def\Radius{0.5}
    \path
      (-\Radius,0) coordinate (A)
      -- coordinate (M)
      (\Radius, 0) coordinate (B)
      (M) +(60:\Radius) coordinate (C)
      +(120:\Radius) coordinate (D)
    ;
    \draw[blue]
      (3,4) arc(0:180:\Radius) -- cycle
    ;
    \node[blue] at (2.5, 5) {\(g\)};

    \draw[violet]
      (8,4) arc(0:180:\Radius) -- cycle
    ;
    \node[violet] at (7.5, 5) {\(g\)};

    \draw[red, thick, {Rays[n=2]}-{Latex[round, red]}] (2,2.4) -- (0,2.4);
    \draw[red, thick,  {Rays[n=2]}-{Latex[round, red]}] (3,2) -- (1,2);
    \node[red] at (2, 1.5) {\(\bar{u}\)};

    \draw[red, thick, {Rays[n=2]}-{Latex[round, red]}] (8,2.4) -- (10,2.4);
    \draw[red, thick,  {Rays[n=2]}-{Latex[round, red]}] (7,2) -- (9,2);
    \node[red] at (8, 1.5) {\(u\)};

    \draw[black, thick] (2,0.2) -- (2,-0.2);
    \node[black] at (2, -0.5) {\(\bar{c}\)};

    \draw[black, thick] (8,0.2) -- (8,-0.2);
    \node[black] at (8, -.5) {\(c\)};

  \end{tikzpicture}
  \caption{The queries that compute $\ou$ and $u$ check how far the period $g$ of $\exp(\hX_\ell)=X\fragmentco{\oc}{c}$ extends within $X$.}\label{fig:u}
  \end{minipage}
  \hspace{1cm}
  \begin{minipage}{0.45\textwidth}
\centering
  \begin{tikzpicture}[scale=0.4]
    \draw[ultra thick, mark=round] (2,3.2) -- (6.5,3.2);
    \draw[ultra thick] (3,3.5) -- (7.5,3.5);
    \draw[ultra thick] (4,3.8) -- (8.5,3.8);

    \def\Radius{0.5}
    \path
      (-\Radius,0) coordinate (A)
      -- coordinate (M)
      (\Radius, 0) coordinate (B)
      (M) +(60:\Radius) coordinate (C)
      +(120:\Radius) coordinate (D)
    ;
    \node[blue] at (5.5, 5) {\(g\)};
    \draw[blue]
      (3,4) arc(0:180:\Radius) -- cycle
    ;
    \draw[blue]
      (4,4) arc(0:180:\Radius) -- cycle
    ;
    \draw[blue]
      (5,4) arc(0:180:\Radius) -- cycle
    ;
    \draw[blue]
      (6,4) arc(0:180:\Radius) -- cycle
    ;
    \draw[blue]
      (7,4) arc(0:180:\Radius) -- cycle
    ;
    \draw[blue]
      (8,4) arc(0:180:\Radius) -- cycle
    ;
    \draw[blue]
    (8.5,4.5) arc(90:180:\Radius) -- (8.5,4)-- cycle;

    \draw[red, thick, {Rays[n=2]}-{Latex[round, red]}] (2,2.4) -- (0,2.4);
    \draw[red, thick,  {Rays[n=2]}-{Latex[round, red]}] (3,2) -- (1,2);
    \node[red] at (2, 1.5) {\(\bar{v}\)};

    \draw[red, thick, {Rays[n=2]}-{Latex[round, red]}] (8.5,2.4) -- (10.5,2.4);
    \draw[red, thick,  {Rays[n=2]}-{Latex[round, red]}] (7.5,2) -- (9.5,2);
    \node[red] at (8.5, 1.5) {\(v\)};

    \node[blue] at (5,1) {\(Y\)};
    \draw[blue] (-1,0.5) -- (11,0.5);

    \draw[black, thick] (2,0.7) -- (2,0.3);
    \node[black] at (2, 0) {\(\bar{a}\)};

    \draw[black, thick] (8.5,0.7) -- (8.5,0.3);
    \node[black] at (8.5, 0) {\(a\)};
  \end{tikzpicture}
  \caption{The thick black lines are the detected occurrences of \(\exp(\hX_\ell)\) in \(Y\); these occurrences start $g$ positions apart and their union is $Y\fragmentco{\oa}{a}$.
  The queries that compute $\ov$ and $v$ check how far the period $g$ of $Y\fragmentco{\oa}{a}$ extends within $Y$.}\label{fig:v}
  \end{minipage}
  \end{figure}

  Let  $a = \oa+(s-1)\cdot g + |\exp(\hX_\ell)|$ be the position of $Y$ immediately following the rightmost occurrence of $\exp(\hX_\ell)$ captured by $V$.
  We ask the following LCE queries; see Figures~\ref{fig:u}~and~\ref{fig:v}:
  \begin{align*}
    \ou &:= \overline{\LCE}_{X}(\oc, \oc+g),&
    u &:= \LCE_X(c, c-g),\\
    \ov &:= \overline{\LCE}_{Y}(\oa, \oa+g),&
    v &:= \LCE_Y(a, a-g).
  \end{align*}

  \subparagraph*{Case 1: \(\ou= \oc\) and \(u= |X|-c\).}
  In this case, the period $g$ of $\exp(\hX_\ell)=X\fragmentco{\oc}{c}$ extends to the entire~$X$.
  Therefore, an occurrence of $\exp(\hX_\ell)$ contained in $Y\fragmentco{\oa}{a}$ extends to an occurrence of $X$ in $Y$ if and only if the latter is contained in $Y\fragmentco{\oa-\ov}{a+v}$, which is the maximal extension of $Y\fragmentco{\oa}{a}$ with period $g$.
  The starting positions of these occurrences of $X$ form an arithmetic progression:
  \[\left\{\oa - \oc + i\cdot g : i\in \left[\left\lceil\tfrac{\max(0,\ou-\ov)}{g}\right\rceil\dd s-\left\lceil\tfrac{\max(0,u-v)}{g}\right\rceil\right)\right\}.\]

  \subparagraph*{Case 2: \(\ou< \oc\).}
  Then, the period $g$ of $\exp(\hX_\ell)=X\fragmentco{\oc}{c}$ breaks at position $X\position{\oc-\ou-1}$.
  The corresponding position within any occurrence of $X$ in $Y$ must also break the period.
  The period of $Y\fragmentco{\oa}{a}$ breaks at position $Y\position{\oa-\ov-1}$.
  Thus, whenoever $X\fragmentco{\oc}{c}$ is aligned within $Y\fragmentco{\oa}{a}$, then $X\position{\oc-\ou-1}$ must be aligned against $Y\position{\oa-\ov-1}$, and the only candidate occurrence is $Y\fragmentco{\oa-\oc+\ou-\ov}{\oa-\oc+\ou-\ov+|X|}$.
  If this is not a valid fragment of $Y$ (due to out-of-bounds indices), then there is no induced occurrence; otherwise, the single candidate can be verified with one LCE query.

  \subparagraph*{Case 3: \(u < |X|-c\).\protect\footnote{If Cases 2 and 3 hold simultaneously, either procedure can be used.}}
  This case is symmetric to the previous one up to reversing $X$ and $Y$.
  The period $g$ of $\exp(\hX_\ell)=X\fragmentco{\oc}{c}$ breaks at position $X\position{c+u}$, whereas the period $g$ of $Y\fragmentco{\oa}{a}$ breaks at position $Y\position{a+v}$.
  The only candidate occurrence is \(Y\fragmentco{a-c+v-u}{a-c+v-u+|X|}\), where these two positions are aligned; if this candidate is a valid fragment, it can be verified with a single LCE query.

  \subparagraph*{Summary.} We were able to construct the set of occurrences of \(X\) in \(Y\) with up to five LCE queries for each of the $\Oh(1)$ arithmetic progressions representing $\Occ(\hX_\ell,Y'_\ell)$.
  Each of these queries can be answered in \(\bigO(r)\) time using Proposition~\ref{prp:lce}.
  Therefore, our verification algorithm has a runtime of \(\bigO(r)\).

\begin{lemma}\label{lem:confirm}
  Given \(\Occ(\hX_\ell, Y'_\ell)\) as a set of \(\bigO(1)\) arithmetic progressions with differences at most $|\hX_\ell|$, we can compute \(\Occ(X,Y)\) in time \(\bigO(r)\).\lipicsEnd
\end{lemma}

We note that Lemma~\ref{lem:confirm} is similar to \cite[Lemma 1.14 (b)]{KRRW23}. Both answer a set of LCE queries in a periodic string while only interested in the indices of the largest results. Since the LCE queries in our compressed setting are slower than in theirs, we get a larger runtime in total.

\bibliographystyle{alphaurl}
\bibliography{paper}

\end{document}